\documentclass[aps,pra,onecolumn,superscriptaddress,10pt]{revtex4-2}

\usepackage[margin=3cm]{geometry}
\usepackage{amsmath, amsthm, amsfonts, amssymb, bm, bbm}
\usepackage{qcircuit}
\usepackage{hyperref, cleveref}
\hypersetup{
    colorlinks,
    linkcolor={blue},
    citecolor={blue},
    urlcolor={blue}
}

\newcommand{\cgate}[1]{*+<.6em>{#1} \POS ="i","i"+UR;"i"+UL **\dir{-};"i"+DL **\dir{-};"i"+DR **\dir{-};"i"+UR **\dir{-},"i" \cw}

\theoremstyle{plain}
\newtheorem{theorem}{Theorem}
\newtheorem{lemma}[theorem]{Lemma}

\theoremstyle{definition}

\DeclareMathOperator{\tr}{Tr}

\newcommand{\bb}[1]{\ensuremath{\mathbb{#1}}}
\newcommand{\mc}[1]{\ensuremath{\mathcal{#1}}}
\newcommand{\mf}[1]{\ensuremath{\mathfrak{#1}}}

\makeatletter
\def\ft#1{\def\tempa{#1}\futurelet\next\ft@i}
\def\ft@i{\ifx\next\bgroup\expandafter\ft@ii\else\expandafter\ft@end\fi}
\def\ft@ii#1{{\ensuremath{\mf{F}_{\tempa}\!\left(#1\right)}}}
\def\ft@end{{\ensuremath{\mf{F}_{\tempa}}}}
\makeatother

\newcommand{\ket}[1]{\ensuremath{| #1 \rangle}}
\newcommand{\bra}[1]{\ensuremath{\langle #1 |}}

\makeatletter
\def\ketbra#1{\def\tempa{#1}\futurelet\next\ketbra@i}
\def\ketbra@i{\ifx\next\bgroup\expandafter\ketbra@ii\else\expandafter\ketbra@end\fi}
\def\ketbra@ii#1{\left| \tempa \middle\rangle\!\middle\langle #1 \right|}
\def\ketbra@end{\left| \tempa \middle\rangle\!\middle\langle \tempa \right|}
\makeatother

\makeatletter
\def\braket#1{\def\tempa{#1}\futurelet\next\braket@i}
\def\braket@i{\ifx\next\bgroup\expandafter\braket@ii\else\expandafter\braket@end\fi}
\def\braket@ii#1{\left\langle \tempa \middle| #1 \right\rangle}
\def\braket@end{\left\langle \tempa \middle| \tempa \right\rangle}
\makeatother

\newcommand{\dbra}[1]{\ensuremath{\langle\!\langle #1|}}
\newcommand{\dket}[1]{\ensuremath{|#1\rangle\!\rangle}}
\newcommand{\cbra}[1]{\ensuremath{(#1|}}
\newcommand{\cket}[1]{\ensuremath{|#1)}}

\makeatletter
\def\cketbra#1{\def\tempa{#1}\futurelet\next\cketbra@i}
\def\cketbra@i{\ifx\next\bgroup\expandafter\cketbra@ii\else\expandafter\cketbra@end\fi}
\def\cketbra@ii#1{\left| \tempa \middle)\!( #1 \right|}
\def\cketbra@end{\left| \tempa \middle)\!\middle( \tempa \right|}
\makeatother

\makeatletter
\def\cbraket#1{\def\tempa{#1}\futurelet\next\cbraket@i}
\def\cbraket@i{\ifx\next\bgroup\expandafter\cbraket@ii\else\expandafter\cbraket@end\fi}
\def\cbraket@ii#1{\left( \tempa \middle| #1 \right)}
\def\cbraket@end{\left( \tempa \middle| \tempa \right)}
\makeatother

\makeatletter
\def\dketbra#1{\def\tempa{#1}\futurelet\next\dketbra@i}
\def\dketbra@i{\ifx\next\bgroup\expandafter\dketbra@ii\else\expandafter\dketbra@end\fi}
\def\dketbra@ii#1{| \tempa \rangle\!\rangle\!\langle\!\langle #1 |}
\def\dketbra@end{| \tempa \rangle\!\rangle\!\langle\!\langle \tempa |}
\makeatother

\makeatletter
\def\dbraket#1{\def\tempa{#1}\futurelet\next\dbraket@i}
\def\dbraket@i{\ifx\next\bgroup\expandafter\dbraket@ii\else\expandafter\dbraket@end\fi}
\def\dbraket@ii#1{\langle\!\langle \tempa | #1 \rangle\!\rangle}
\def\dbraket@end{\langle\!\langle \tempa | \tempa \rangle\!\rangle}
\makeatother

\def\CJ/{Choi-Jamio{\l}kowski}

\begin{document}

\title{Randomized compiling for subsystem measurements}
\author{Stefanie J. Beale}
\affiliation{Keysight Technologies Canada, Kanata, ON K2K 2W5, Canada}
\affiliation{Institute for Quantum Computing, University of Waterloo, Waterloo, Ontario N2L 3G1, Canada}
\affiliation{Department of Physics and Astronomy, University of Waterloo, Waterloo, Ontario N2L 3G1, Canada}
\author{Joel J. Wallman}
\affiliation{Keysight Technologies Canada, Kanata, ON K2K 2W5, Canada}
\affiliation{Institute for Quantum Computing, University of Waterloo, Waterloo, Ontario N2L 3G1, Canada}
\affiliation{Department of Applied Mathematics, University of Waterloo, Waterloo, Ontario N2L 3G1, Canada}

\date{\today}

\begin{abstract}
Measurements are a vital part of any quantum computation, whether as a final step to retrieve results, as an intermediate step to inform subsequent operations, or as part of the computation itself (as in measurement-based quantum computing).
However, measurements, like any aspect of a quantum system, are highly error-prone and difficult to model.
In this paper, we introduce a new technique based on randomized compiling to transform errors in measurements into a simple form that removes particularly harmful effects and is also easy to analyze.
In particular, we show that our technique reduces generic errors in a computational basis measurement to act like a confusion matrix, i.e. to report the incorrect outcome with some probability, and as a stochastic channel that is independent of the measurement outcome on any unmeasured qudits in the system.
We further explore the impact of errors on indirect measurements and demonstrate that a simple and realistic noise model can cause errors that are harmful and difficult to model.
Applying our technique in conjunction with randomized compiling to an indirect measurement undergoing this noise results in an effective noise which is easy to model and mitigate.
\end{abstract}

\maketitle

\section{Introduction}

Measurements are the sole method of retrieving information from quantum systems; as such, they are integral to any quantum computation.
In circuit-based models of quantum computation, most measurements are destructive, terminal measurements in the computational basis, used at the end of a circuit to retrieve results.
Noise on these measurements is typically modeled by assuming that the measurement outcomes are permuted according to a \textit{confusion matrix}, that is, a matrix containing probabilities of observing a given outcome given that another outcome occurred.
When noise can accurately be modeled by a confusion matrix, it can be mitigated in post-processing with an overhead that, under reasonable assumptions, scales polynomially with the number of qubits~\cite{Bravyi2021}.

However, measurement outcomes can also be used to choose subsequent operations as in quantum error correction~\cite{gottesman2009introduction} and measurement-based quantum computation~\cite{Raussendorf2001}.
For these applications, the assumption that measurements are terminal and destructive is insufficient, and the noise model that applies a confusion matrix is not well motivated.
The effect of realistic noise when the measured state persists after measurement can be harmful and difficult to model.
This same problem of realistic errors being difficult to model also arises for unitary channels.
For such channels, the effects of noise can be reduced to a stochastic model which is easier to predict and mitigate or correct by applying randomized compiling (RC) \cite{Wallman2016a}.

In this paper, we generalize RC to tailor the effects of measurement errors to a stochastic form that is easier to model.
Like RC for unitary channels, our generalization does not introduce significant overhead in number of shots or number of gates.
We show that our protocol reduces fully general noise on a destructive computational basis measurement to a channel which acts stochastically on the unmeasured qudits in a way that is independent of the measurement outcome, and permutes the outcomes observed on the measured qudits.
We further show how harmful noise can occur from a simple and realistic perturbation of an indirect measurement and that our protocol reduces this noise to a more desirable form.

We begin in \cref{sec:preliminaries} by reviewing basic concepts in quantum computing, including laying out our notation for quantum states, operations and measurements, reviewing the effects of twirling on a noisy quantum operation, and defining the concept of unnormalized stochastic channels to cover the form of noise returned by our technique.
In \cref{sec:indMeasEx} we give an explicit example where a simple and well-motivated noise model on a standard implementation of an indirect measurement results in noise that does not follow the standard assumptions on measurement noise, and which is harmful and difficult to model when subsequent operations are applied.
In \cref{sec:rcMeas} we show how randomized compiling can be applied to non-destructive measurements to tailor harmful errors into a form which is easier to model and less harmful.
\Cref{sec:HigherRankMeasurements} generalizes \cref{sec:indMeasEx} to higher rank measurements, showing how randomized compiling can be applied to indirect measurements to tailor noise effectively, and we return to the example from \cref{sec:indMeasEx} to show the effects of randomized compiling explicitly.
We include in an appendix an explicit derivation of a noisy implementation of a general indirect measurement of a Weyl operation which shows that the harmful noise which is removed by randomized compiling is realistic for noise on higher dimensional systems and for a common subclass of measurements.

\section{Preliminaries}\label{sec:preliminaries}

We begin by defining some notation and reviewing basic concepts we use to represent quantum channels acting on $n$ qudits of dimension $d$.
Throughout this paper, we denote the uniform average of a function $f$ over a set $\bb{S}$ by
\begin{align}
    \bb{E}_{s \in \bb{S}} f(s) = \frac{1}{|\bb{S}|} \sum_{s \in \bb{S}} f(s)
\end{align}
and a tensor product of powers of an operator by
\begin{align}\label{eq:VectorPower}
A^a = \otimes_{i\in \bb{Z}_n} A^{a_i}
\end{align}
where $a \in \bb{R}^n$

\subsection{Quantum states}\label{sec:QuantumStates}

A qudit of dimension $d$ is a quantum system whose state can be described by an element of the Hilbert space $\bb{H}_d = \bb{B}(\bb{C}^d)$, that is, the set of bounded linear operators from $\bb{C}^d$ to itself.
To be a valid description, the state must be a density operator, that is, a positive semi-definite operator with unit trace.
It is convenient to represent quantum states as vectors of expansion coefficients.
Relative to a fixed trace-orthogonal basis $\{B_j : j \in \bb{Z}_d^2\}$ of $\bb{H}_d$ and the canonical orthonormal basis $\{e_j : j \in \bb{Z}_{d^2}\}$ of $\bb{C}^{d^2}$, we define the vectorization map to be the function
\begin{align}
    \dket{} : \bb{H}_d \to \bb{C}^{d^2} :: \dket{M} = \sum_{j \in \bb{Z}_{d^2}} \frac{\tr B_j^\dagger M}{\sqrt{\tr B_j^\dagger B_j}} e_j.
\end{align}
Note that we explicitly include the normalization factors rather than defining the map relative to a trace-orthonormal basis.
We can translate expectation values into standard vector inner products by defining $\dbra{\cdot} = \dket{\cdot}^\dagger$ and using the usual shorthand $\dbraket{A}{B} = \dbra{A}\dket{B}$, so that
\begin{align}
    \dbraket{A}{B} = \tr A^\dagger B
\end{align}
for all $A, B \in \bb{H}_d$.
For example, one can readily see that the normalization factors are such that 
\begin{align}
\dket{B_j} = \sqrt{\tr B_j^\dagger B_j} e_j,    
\end{align}
and so the trace-orthogonality relations between the $B_j$ are equivalent to orthogonality relations between the $\dket{B_j}$.
We can extend the above definitions to $n$ qudits by setting $d \to d^n$, where we can construct bases for $\bb{H}_{d^n}$ by taking the $n$-fold tensor product of a basis for $\bb{H}_d$.
We will frequently work with the $n$-qudit computational basis states and so define
\begin{align}
    \cket{} : \bb{Z}_d^n \to \bb{C}^{d^2} :: \cket{j} = \dket{\ketbra{j}}.
\end{align}
Note that because the $\cket$ stores a computational basis state, it can be treated as a classical dit or as a qudit with a restricted state space that mimics a classical dit.
For clarity, note that if we conjugate a state $\rho$ by a Kraus operator $L = K \otimes \ket{k}$, we get the state
\begin{align}
    L \rho L^\dagger = M \rho M^\dagger \otimes \ketbra{k}.
\end{align}
Vectorizing the result then gives
\begin{align}
    \dket{L \rho L^\dagger} = \dket{M \rho M^\dagger} \otimes \cket{k},
\end{align}
so that when we switch between vectorized and unvectorized statements we replace $\cket{k}$ by $\ketbra{k}$.

\subsection{Quantum operations}\label{sec:QuantumOperations}

We now introduce quantum operations, where for clarity we will use capital Roman letters (and $\pi$ for projectors) to denote operators, calligraphic font to denote quantum operations, and $\Theta(\mc{Q})$ to denote the noisy implementation of a quantum operation $\mc{Q}$.
A quantum operation is a map from quantum states to quantum states and so is an element of the set $\bb{H}_{d, e}$ of bounded linear operators from $\bb{H}_d$ to $\bb{H}_e$.
As such, one can write any $\mc{Q} \in \bb{H}_{d, e}$ as
\begin{align}
    \mc{Q} = \sum_j \dketbra{A_j}{B_j}
\end{align}
for some $\{A_j\} \subset \bb{H}_e$ and $\{B_j\} \subset \bb{H}_d$.
When $\mc{Q}$ is a completely positive map, there exist Kraus operators $\{K_k\} \subset \bb{C}^{e \times d}$ such that for any trace-orthogonal basis $\{B_j : j \in \bb{Z}_d^2\}$ of $\bb{H}_d$, we have
\begin{align}
    \mc{Q} = \sum_{j,k} \dketbra{K_k B_j K_k^\dagger}{B_j}.
    \label{eq:KrausSuperop}
\end{align}
As a special case, for any unitary operator $U \in \bb{H}_d$, we define the corresponding channel
\begin{align}\label{eq:UnitaryChannel}
    \mc{U} = \sum_{j,k} \dketbra{U B_j U^\dagger}{B_j}.
\end{align}

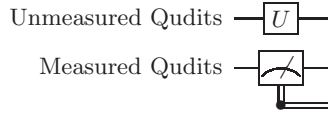
\begin{figure}[t!]
\begin{align*}
\Qcircuit @C=1em @R=.7em {
\lstick{\mbox{Unmeasured Qudits}} & \gate{U} & \qw \\
\lstick{\mbox{Measured Qudits}} & \meter\cwx[1]  & \qw \\
& \control & \cw
}
\end{align*}
\caption{A circuit diagram showing a subsystem measurement, which consists of a computational basis measurement of a subset of $m$ qudits of dimension $d$ and some evolution $U$ on the rest of the system.
The corresponding channel is given in \cref{eq:SubsystemMeasurement}.
}
\label{fig:meas}
\end{figure}

While unitary channels are fundamental to quantum computation, they are never implemented perfectly. 
In full generality, an imperfect implementation of a unitary channel can be a completely arbitrary channel.
However, a common class of errors that are analyzed are stochastic channels~\cite{Graydon2022}, which admit a Kraus operator description $\{K_j\}$ such that $K_0 = p I$ for some $p \in (0, 1]$ and 
\begin{align}
    \tr K_j^\dagger K_k = \delta_{j, k} \tr K_j^\dagger K_j.
\end{align}
Stochastic channels are easier to analyze because the worst-case error rate is approximately equal to the average error rate, and are accurate as a model because general errors can be converted into stochastic channels using randomized compiling~\cite{Wallman2016a}.
Although Ref.~\cite{Graydon2022} defined stochastic channels to be trace-preserving maps, the proof that the twirled channel
\begin{align}
    \bb{E}_{G \in \bb{G}} \mc{G}^\dagger \Lambda \mc{G}
\end{align}
is a stochastic channel for any channel $\Lambda$ and unitary 1-design $\bb{G}$~\cite{Dankert2009} in \cite[Theorem 3]{Graydon2022} can be trivially generalized to completely positive maps, where the twirled channel admits the same Kraus operator decomposition without the trace-preserving condition.
We refer to such channels as unnormalized stochastic channels.

Another special case of interest is a type of measurement that we refer to as \emph{subsystem measurements}, wherein a subset of $m$ qudits are non-destructively measured in the computational basis.
Formally, let $\bb{H}_{d^n} \equiv \bb{H}_{\rm rest} \otimes \bb{H}_{\rm meas}$ where $\bb{H}_{\rm meas} \equiv \bb{H}_{d^m}$ is the space of states for the $m$ qudits being measured.
Then a subsystem measurement with side-evolution $U$ applied to the unmeasured qudits corresponds to measuring the $m$ qudits and returning the outcome $k \in \bb{Z}_d^m$ with probability
\begin{align}
    \tr \pi_k \rho,
\end{align}
where $\pi_k = U \otimes \ketbra{k}$, and we allow $U \in \bb{H}_{d^{n-m}}$ to not be the identity to allow other computation to take place during the measurement.
As the measurement is a projective rank 1 measurement, the post-measurement state conditioned on the outcome $k$ being observed is given by L\"{u}der's rule,
\begin{align}
    \frac{\pi_k \rho \pi_k^\dagger}{\tr \pi_k \rho \pi_k^\dagger}.
\end{align}
We can return the outcome $k$ by ``creating'' a new $m$-qudit system that always has a definite state.
This gives us the Kraus operators $\pi_k \otimes \ket{k}$ and so we obtain the overall channel
\begin{align}\label{eq:SubsystemMeasurement}
    \mc{M}_{U, \bb{H}_{\rm meas}} = \mc{U} \otimes \sum_{k \in \bb{Z}_d^m} \cketbra{k} \otimes \cket{k},
\end{align}
where the second tensor factor corresponds to the returned outcome.
Note that this channel does not violate the no-cloning theorem because it is only cloning orthogonal states.

A subsystem measurement includes the case where $m = n$, that is, where the entire system is measured.
While there are in principle more general types of measurements, in practice more complex measurements are typically implemented by applying an entangling unitary and then performing a subsystem measurement.
We will return to this in more detail in \cref{sec:HigherRankMeasurements}.

Finally, we define uniform stochastic instruments, which correspond to a specific type of imperfect subsystem measurement.
A uniform stochastic instrument is a channel that can be written as
\begin{align}\label{eq:UniformStochasticMeasurement}
    \sum_{a, b, k \in \bb{Z}_d^m} \mc{U} \mc{T}_{a, b} \otimes \cketbra{k + b}{k + a} \otimes \cket{k} 
\end{align}
where $\mc{T}_{a, b}$ is an unnormalized stochastic channel that acts only on $\bb{H}_{\rm rest}$.
A uniform stochastic instrument is a particularly simple error model in that there are only three error modes:
\begin{enumerate}
    \item A state may be misreported, that is, the system is in the state $\cket{k+a}$ but is reported as $\cket{k}$;
    \item The measured systems are left in the wrong computational basis state, that is, they are left in the state $\cket{k+b}$ when the outcome $k$ is reported; and
    \item A stochastic error is applied to the unmeasured systems that is independent of the observed outcome.
\end{enumerate}
We refer to such channels as ``uniform'' because the error applied to the rest of the system is uniform over the observed outcome.
Moreover, the measured and unmeasured qudits are disentangled by such a measurement even if they were entangled before the measurement.
We can see this by noting that the post-measurement state is a linear combination of tensor products of arbitrary states of the unmeasured qudits with pure states of the measured qudits.
The states of the measured and unmeasured qudits after the measurement are typically classically correlated, unless only one $\mc{T}_{a, b}$ is nonzero in \cref{eq:UniformStochasticMeasurement}.
Note that the probability of observing the outcome $k$ and leaving the system in the state $\cket{k+b}$ given an initial state $\rho$ is
\begin{align}
    \sum_{a \in \bb{Z}_d^m} \tr \mc{T}_{a, b} \otimes \cbra{k+a}(\rho),
\end{align}
and so is determined by the unnormalized stochastic channels.
Generally, some outcome will always occur and so
\begin{align}
    \sum_{a, b \in \bb{Z}_d^m} \mc{T}_{a, b}
\end{align}
is a trace-preserving map.

\subsection{Cyclic groups}\label{sec:CyclicGroups}

Many of the quantum operations we use are Weyl operators, which are defined and analyzed using the following properties of cyclic groups.
A more detailed treatment can be found in, e.g., Ref.~\cite{fulton2011representation}.

The cyclic group of order $d$ is the set $\bb{Z}_d $ equipped with addition modulo $d$.
The irreducible representations of $\bb{Z}_d$ are the $d$ functions
\begin{align}\label{eq:CyclicCharacters}
    \chi_a : \bb{Z}_d \to \bb{C} :: \chi_a(b) = \exp(2 \pi i a b / d),
\end{align}
where the value of $d$ is implicit from the context.
We can extend the above representations to obtain the $d^n$ irreducible representations of the Cartesian product $\bb{Z}_d^n$ by
\begin{align}
\chi_a : \bb{Z}_d^n \to \bb{C} :: \chi_a(b) = \prod_{i\in \bb{Z}_n} \chi_{a_i}(b_i).
\end{align}
These representations are all characters and thus obey Schur's orthogonality relations,
\begin{align}\label{eq:CyclicOrthogonality}
    \bb{E}_{a \in \bb{Z}_d^n} \chi_j(a) \bar{\chi}_k(a) = \delta_{j,k}.
\end{align}

\subsection{Weyl operators}

We now introduce Weyl operators, which generalize the familiar multi-qubit Pauli operators to qudits of dimension $d \geq 2$.
Like the Pauli group, the projective Weyl group is a 1-design which forms an orthogonal basis for the operator space.
As with the Pauli operators, we begin by defining the $X$ and $Z$ operators, which for dimension $d$ are
\begin{align}
    X &\equiv \sum_{j\in\bb{Z}_d} \ketbra{j+1}{j} \notag\\
    Z &\equiv \sum_{j\in\bb{Z}_d} \chi_1(j) \ketbra{j},
\end{align}
where $d$ will be implicit from the context.
Noting that
\begin{align}
    X^x &= \sum_{j \in \bb{Z}_d^n} \ketbra{j + x}{j} \\
    Z^z &= \sum_{j \in \bb{Z}_d^n} \chi_z(j) \ketbra{j},\label{eq:ZPower}
\end{align}
we can readily verify the braiding relation
\begin{align}\label{eq:braiding}
    Z^z X^x = \chi_z(x) X^x Z^z,
\end{align}
and, using \cref{eq:CyclicOrthogonality}, the orthogonality relation
\begin{align}\label{eq:WeylOrthogonality}
    \tr [(Z^z X^x)^\dagger (Z^\alpha X^\beta)] = d^n \delta_{x, \beta} \delta_{z, \alpha}. 
\end{align}
We can also use \cref{eq:CyclicOrthogonality} to invert \cref{eq:ZPower} and obtain
\begin{align}\label{eq:ProjectorExpansion}
    \ketbra{j} = \bb{E}_{z \in \bb{Z}_d^n} \bar{\chi}_j(z) Z^z.
\end{align}
The projective Weyl group is the set
\begin{align}
    \bb{PW}_{d,n} \equiv \{X^x Z^z:x,z\in\bb{Z}_d^n\},
\end{align}
which is a trace-orthogonal basis for $\bb{H}_d$.

\section{Example: indirect measurements}\label{sec:indMeasEx}

\begin{figure}[t!]
\begin{align*}
\Qcircuit @C=1em @R=.7em {
\lstick{\mbox{Data qubit: } \rho} & \ctrl{1} & \qw  & \qw \\
\lstick{\mbox{Readout qubit: } \bra{0}} & \targ & \meter & \cw
}
\end{align*}
\caption{Circuit diagram of an indirect measurement. A data qubit in an initial state $\rho$ is coupled to a readout qubit initialized in the state $\bra{0}$ by a CNOT gate. The readout qubit is then measured in the computational basis.}
\label{fig:indirectMeas}
\end{figure}
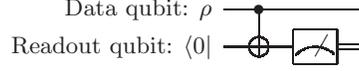

We now illustrate the complexity of modelling errors in subsystem measurements without randomized compiling by considering an \emph{indirect} subsystem measurement, wherein the measurement is performed by coupling the system be measured to another quantum system that can be directly measured.
For simplicity, we will restrict to an indirect computational basis measurement of a single data qubit by coupling it to another readout qubit as illustrated in \cref{fig:indirectMeas}.
Typically, $V$ would correspond to performing a controlled not gate, which can be written as
\begin{align}
    V = \ketbra{0} \otimes I + e^{i \pi / 2} \ketbra{1} \otimes e^{-i \pi X / 2}.
\end{align}
The Kraus operators describing the action of an ideal implementation of \cref{fig:indirectMeas} on the data qubit are
\begin{align}
    K_j &= (I_2 \otimes \ketbra{j}{j})V(I_2 \otimes \ket{0}) = \ketbra{j}{j} \otimes \ket{j},
\end{align}
for $j = 0,1$, which corresponds to an ideal computational basis measurement of the data qubit as desired, where the second tensor factor is the observed measurement outcome.

Now suppose that the controlled-not gate is implemented with an over-rotation error, which we will take to be
\begin{align}\label{eq:Vdef}
    V_\phi = \ketbra{0} \otimes I + \ketbra{1} \otimes e^{-i \phi X},
\end{align}
where $\phi = \pi/2$ corresponds to the ideal case and we ignore the $e^{i \pi / 2}$ factor as it only appears as an overall phase on $K_1$.
This results in the Kraus operators
\begin{align}\label{eq:NoisyKraus}
    \tilde{K}_0 &= \ketbra{0}{0} \otimes \ket{0} + \cos \phi \ketbra{1}{1} \otimes \ket{0} \notag\\
    \tilde{K}_1 &= - i \sin \phi \ketbra{1}{1} \otimes \ket{1}.
\end{align}
While this is an extremely simple error model representing just one likely error in an implementation, it already complicates the behavior.
Consider an initial state $\rho = \ketbra{+}$ for the data qubit.
The unnormalized ideal output after observing the outcome $\ket{0}$ is $\tfrac{1}{2}\ketbra{0}$, where the trace of the unnormalized output gives the probability with which it occurs.
However, with the Kraus operators in \cref{eq:NoisyKraus} we instead obtain
\begin{align}
\tilde{K}_0 \ketbra{+} \tilde{K}_0^\dagger = \tfrac{1}{2}\left(\ketbra{0}{0} + \cos\phi \ketbra{0}{1} + \cos \phi \ketbra{1}{0} + \cos^2\phi \ketbra{1}{1}\right) \otimes \ketbra{0}.
\label{eq:noisyK0Zmeas}
\end{align}
This incorrect state cannot be modelled by a confusion matrix, as it is not a mixture of $\ketbra{0}$ and $\ketbra{1}$.
We will return to this example in \cref{sec:ReturnOfTheIndMeasEx} to show how randomized compiling results in a simpler error model.

\section{Randomized compiling for non-destructive quantum measurements}\label{sec:rcMeas}

In this section, we introduce randomized compiling for non-destructive quantum measurements and analyze the impact on a general noisy measurement.
A general noisy implementation of a subsystem measurement $\mc{M}_{U, \bb{H}_{\rm meas}}$ as illustrated in \cref{fig:meas} can be expressed as
\begin{align}
    \Theta(\mc{M}_{U, \bb{H}_{\rm meas}}) = \sum_{k \in \bb{Z}_d^n} \mc{M}_k \otimes \cket{k},
    \label{eq:generalNoisyMeas}
\end{align}
where the $\mc{M}_k$ are completely positive maps from $\bb{H}$ to itself and the $\cket{k}$ corresponds to the outcome. 
That is, a known classical outcome is returned but the corresponding evolution of the system can be arbitrary provided that the state of the system conditioned on an outcome is a valid quantum state.
As some outcome is always observed, $\sum_k \mc{M}_k$ is a trace-preserving map, however, we do not need to formally impose this restriction.

Our goal is to transform a noisy implementation of a subsystem measurement of the form in \cref{eq:generalNoisyMeas} to be a uniform stochastic measurement as in \cref{eq:UniformStochasticMeasurement} by adding gates drawn uniformly at random from the following three sets as illustrated in \cref{fig:measRC}:
\begin{enumerate}
    \item diagonal gates on the measured qudits before and after the measurement to diagonalize the measurement;
    \item basis permutations on the measured qudits to permute the expected outcome and undo the permutation after the measurement; and
    \item a unitary 1-design on the unmeasured qudits before and after the measurement to reduce the errors on the unmeasured qudits to stochastic errors.
\end{enumerate}

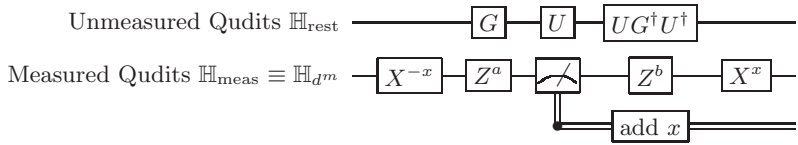
\begin{figure}[t!]
\begin{align*}
\Qcircuit @C=1em @R=.7em {
\lstick{\mbox{Unmeasured Qudits }\bb{H}_{\rm rest}} & \qw & \gate{G} & \gate{U} &  \gate{U G^\dagger U^\dagger} & \qw & \qw \\
\lstick{\mbox{Measured Qudits }\bb{H}_{\rm{meas}}\equiv \bb{H}_{d^m}} & \gate{X^{-x}} & \gate{Z^a} & \meter\cwx[1]  & \gate{Z^b} & \gate{X^x} & \qw  \\
 & & & \control & \cgate{\mbox{add } x}  & \cw & \cw
}
\end{align*}
\caption{A circuit diagram showing a randomly compiled computational implementation of \cref{fig:meas}, where $a, b, x \in \bb{Z}_d^m$ are $m$-digit strings with entries modulo $d$ and $G$ is an element of a unitary 1-design. In \cref{thm:RCMeas}, we prove that the uniform average over $a,b,x,$ and $G$ is a uniform stochastic instrument.}
\label{fig:measRC}
\end{figure}

For clarity, we will first prove the well-known fact that the average effect of applying uniformly random diagonal gates as introduced before and after the measurement produces a dephasing channel before analyzing the effect of all three steps together.
Throughout this paper, we will assume that the implementations of the random gates are ideal.
This assumption can be trivially generalized to errors where the implementation of each random gate $U$ can be written as $\mc{L}\mc{U}\mc{R}$ for some constant maps $\mc{L}$ and $\mc{R}$ by redefining $\Theta(\mc{M}_{U, \bb{H}_{\rm meas}}) \to \mc{R} \Theta(\mc{M}_{U, \bb{H}_{\rm meas}}) \mc{L}$ and similarly for other imperfectly implemented operations~\cite{Wallman2018}, which would correspond to the ``hard'' cycles in randomized compiling~\cite{Wallman2016a}.
To both minimize the number of additional physical gates that are introduced and to justify the above redefinition of the noisy implementations of measurements and other hard cycles, adjacent random gates should be compiled into each other and into neighbouring single-qudit gates from the bare circuit where possible.
As such, when measurements are randomly compiled in addition to other hard cycles following Ref.~\cite{Wallman2016a}, the number of physical gates to be implemented does not change.
Moreover, in the limit where random gates are chosen independently for each execution of a circuit, no additional executions are required~\cite{Granade2014}.

\begin{lemma}\label{lem:Dephasing}
For any positive integers $d$ and $n$ with $d \geq 2$,
\begin{align*}
    \bb{E}_{a \in \bb{Z}_d^n} \mc{Z}^a = \sum_{z \in \bb{Z}_d^n} \dketbra{Z^z} = \sum_{j \in \bb{Z}_d^n} \cketbra{j}.
\end{align*}
\end{lemma}

\begin{proof}
Using the shorthand in \cref{eq:VectorPower} and \cref{eq:UnitaryChannel} with $\bb{PW}_{d, n}$ as the basis, we have
\begin{align}\label{eq:diagonalChannel}
    \mc{Z}^a &= \sum_{x, z \in \bb{Z}_d^m} \dketbra{Z^a X^x Z^z Z^{-a}}{X^x Z^z} \notag\\
&= \sum_{x, z \in \bb{Z}_d^m} \chi_a(x)\dketbra{X^x Z^z},
\end{align}
where the second equality follows from \cref{eq:braiding}.
Averaging over $a$ and using \cref{eq:CyclicOrthogonality} gives
\begin{align}\label{eq:dephasing}
    \bb{E}_{a \in \bb{Z}_d^n} \mc{Z}^a &= \sum_{x, z \in \bb{Z}_d^n} \dketbra{X^x Z^z}{X^x Z^z} \bb{E}_{a \in \bb{Z}_d^m} \chi_a(x)\notag\\
    &=\sum_{x, z \in \bb{Z}_d^n} \dketbra{X^x Z^z}{X^x Z^z} \bb{E}_{a \in \bb{Z}_d^n} \chi_x(a)\bar{\chi}_0(a)\notag\\
    &= \sum_{z \in \bb{Z}_d^n} \dketbra{Z^z},
\end{align}
establishing the first equality.

Substituting \cref{eq:ProjectorExpansion} into \cref{eq:dephasing} and using linearity gives
\begin{align}\label{eq:dephasing2}
    \bb{E}_{a \in \bb{Z}_d^n} \mc{Z}^a &= \sum_{j, k \in \bb{Z}_d^m} \cketbra{j}{k} \bb{E}_{a \in \bb{Z}_d^m} \chi_j(a) \bar{\chi}_k(a) \notag\\
    &= \sum_{j \in \bb{Z}_d^m} \cketbra{j},
\end{align}
establishing the second equality.
\end{proof}

Having analyzed the effect of the random diagonal gates, we now prove that averaging over all the random gates in \cref{fig:measRC} produces a uniform stochastic instrument.

\begin{theorem}\label{thm:RCMeas}
If the implementation of a noisy measurement as in \cref{fig:meas} is described by the channel \cref{eq:generalNoisyMeas} and the additional gates in \cref{fig:measRC} are ideal, then the average of \cref{fig:measRC} over $a, b, x \in \bb{Z}_d^m$ and $G$ over a unitary 1-design $\bb{G}$ is a uniform stochastic instrument.
\end{theorem}

\begin{proof}
Under the stated assumptions, we can write the channel implemented by \cref{fig:measRC} for fixed $a, b, x \in \bb{Z}_d^m$ and $G \in \bb{G}$ as
\begin{align}\label{eq:fullChannel}
    \Theta_{a, b, x, G}(\mc{M}_{U, \bb{H}_{\rm meas}}) 
    &= \sum_{k \in \bb{Z}_d^m} \Lambda_{k, a, b, x, g} \otimes \cket{k}
\end{align}
where we have relabeled the classical outcomes to account for adding $x$ so that the measurement acts as $\mc{M}_{k-x}\otimes\cket{k}$ and we define
\begin{align}\label{eq:Lambda}
 \Lambda_{k, a, b, x, G}
 &= (\mc{U}\mc{G}^\dagger\mc{U}^\dagger \otimes \mc{X}^x \mc{Z}^b) \mc{M}_{k - x} (\mc{G} \otimes \mc{Z}^a\mc{X}^{-x} ).
\end{align}
The average over $a, b, x \in \bb{Z}_d^m$ and $G\in\bb{G}$ is then
\begin{align}
    \Theta_{\rm RC}(\mc{M}_{U, \bb{H}_{\rm meas}}) = \bb{E}_{a,b,x \in \bb{Z}_d^m, G \in \bb{G}} \Theta_{a, b, x, G} (\mc{M}_{U, \bb{H}_{\rm meas}}) .
\end{align}
We will evaluate this average by first averaging over the diagonal gates, that is, over $a,b \in \bb{Z}_d^m$, then over the permutation gates, that is, over $x \in \bb{Z}_d^m$, and finally over the unitary 1-design, that is, over $G \in \bb{G}$.
By \cref{lem:Dephasing}, the independent uniform average over $a, b \in \bb{Z}_d^m$ gives
\begin{align}
    \bb{E}_{a, b} \Lambda_{k, a, b, x, G} = \sum_{\alpha, \beta \in \bb{Z}_d^m} (\mc{U}\mc{G}^\dagger\mc{U}^\dagger \otimes \mc{X}^x \cketbra{\beta}) \mc{M}_{k - x} (\mc{G} \otimes \cketbra{\alpha}\mc{X}^{-x} ).
\end{align}
Now note that the $\mc{X}$ channels permute the ideal projectors, since
\begin{align}
X^x \ketbra{j} X^{-x} 
&= \ketbra{j + x}.
\end{align}
Therefore for any fixed $x$ we have $\mc{X}^x \cket{j} = \cket{j + x}$ and so averaging $\Lambda_{k, a, b, x, G}$ over $a, b, x \in \bb{Z}_d^m$ for fixed values of $k$ and $G$ gives
\begin{align}\label{eq:postDephasing}
    \bb{E}_{a, b, x \in \bb{Z}_d^m} \Lambda_{k, a, b, x, G} 
    &= \sum_{\alpha, \beta \in \bb{Z}_d^m} \bb{E}_{x \in \bb{Z}_d^m} (\mc{U}\mc{G}^\dagger\mc{U}^\dagger \otimes \cketbra{\beta + x}{\beta}) \mc{M}_{k - x} (\mc{G} \otimes \cketbra{\alpha}{\alpha + x}) \notag\\ 
    &= \sum_{\alpha, \beta \in \bb{Z}_d^m} \mc{U}\mc{G}^\dagger \mc{M}_{k, \alpha, \beta} \mc{G} \otimes \cketbra{k + \beta}{k + \alpha},
\end{align}
where we relabel the terms in the sum by setting $\alpha\rightarrow \alpha+k-x$ and $\beta\rightarrow\beta+k-x$ and define
\begin{align}\label{eq:Mkab}
    \mc{M}_{k,\alpha,\beta} = \bb{E}_{x \in \bb{Z}_d^m} \left(\mc{U}^\dagger \otimes \cbra{k + \beta - x}\right) \mc{M}_{k-x} \left(\mc{I}_{\rm rest} \otimes \cket{k + \alpha - x}\right),
\end{align}
which only acts on $\bb{H}_{\rm rest}$.
Therefore to complete the proof we need to show that $\bb{E}_G\mc{G}^\dagger\mc{M}_{k,a,b}\mc{G}$ is a stochastic channel that is independent of $k$.
To see that it is independent of $k$, note that if we relabel $x \to k + x$ in \cref{eq:Mkab}, then the right-hand side is manifestly independent of $k$.
As $\mc{M}_{k, \alpha, \beta}$ only acts on $\bb{H}_{\rm rest}$, we set
\begin{align}
    \mc{T}_{\alpha, \beta} = \bb{E}_{G \in \bb{G}} \mc{G}^\dagger \mc{M}_{k, \alpha, \beta} \mc{G} = \bb{E}_{G \in \bb{G}} \mc{G}^\dagger \mc{M}_{0, \alpha, \beta} \mc{G},
\end{align}
which is an unnormalized stochastic channel~\cite{Graydon2022}, completing the proof.
\end{proof}

\section{Effective higher rank measurements}\label{sec:HigherRankMeasurements}

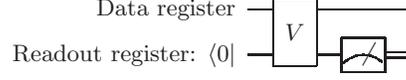
\begin{figure}[t!]
\begin{align*}
\Qcircuit @C=1em @R=.7em {
\lstick{\mbox{Data register}} &  \multigate{1}{V } & \qw & \qw \\
\lstick{\mbox{Readout register: } \bra{0}} &  \ghost{V} & \meter & \cw
}
\end{align*}
\caption{Circuit diagram of a general indirect measurement, which can be used to perform high rank measurements such as syndrome measurements.}
\label{fig:HigherRankMeasurements}
\end{figure}

We now analyze how randomly compiling subsystem measurements can be used to perform effective higher rank measurements with simple error models.
Higher rank measurements are one of the fundamental techniques in quantum error correction as they are used to detect if the system is in some subspace called the code space without disturbing the state of the system if it is in the code space and to determine the correction to apply if it is not in the code space.
Such measurements are performed by a circuit as in \cref{fig:HigherRankMeasurements}, which generalizes \cref{fig:indirectMeas}.
The corresponding Kraus operators are
\begin{align}
    K_k = (I \otimes \bra{k}) V (I \otimes \ket{0}) \otimes \ket{k},
\end{align}
where the $\ket{k}$ represents the observed outcome.
Typically, we choose
\begin{align}\label{eq:ConditionalProjectors}
    V = \sum_{j, k \in \bb{Z}_d^m} \pi_k \otimes \ketbra{k + j}{j}
\end{align}
for a projective measurement $\{\pi_k : k \in \bb{Z}_d^m\}$.
To see that $V$ is unitary, note that for any projective measurement we have $\sum_k \pi_k = I$ and so
\begin{align}
    V^\dagger V = \sum_{j, k \in \bb{Z}_d^m} \pi_k \otimes \ketbra{j} = I \otimes \sum_{j \in \bb{Z}_d^m} \ketbra{j} = I,
\end{align}
and similarly for $V V^\dagger$.
Thus the ideal Kraus operators are
\begin{align}
    K_k = \pi_k \otimes \ket{k},
\end{align}
which is equivalent to performing the ideal projective measurement on the state register and encoding the result in the state of the readout register.
As in \cref{sec:indMeasEx}, simple error models will result in Kraus operators that are not proportional to the ideal projectors and thus leave coherences, arising from cross-terms between the orthogonal subspaces corresponding to different ideal outcomes.
For example, in \cref{sec:appendix-noisy-ind-meas}, we show that an indirect measurement projecting onto the eigenspaces of a Weyl operator with a simple error leaves such coherences. These coherences are assumed to not exist in many analyses of quantum error correction and can cause particularly harmful accumulation of errors as different cospaces and combinations thereof are allowed to interfere in subsequent operations.

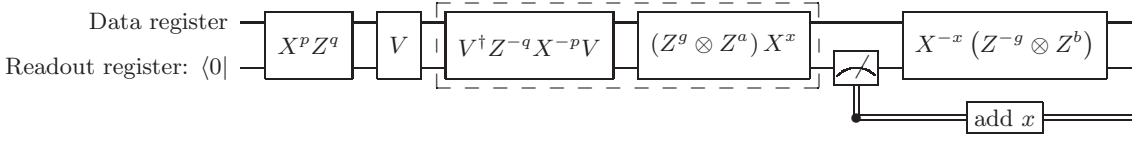
\begin{figure}
\begin{align*}
\Qcircuit @C=1em @R=.7em {
\lstick{\mbox{Data register}} & \multigate{1}{X^p Z^q} & \multigate{1}{V} & \multigate{1}{V^\dagger Z^{-q} X^{-p} V} &\multigate{1}{\left(Z^g \otimes Z^a\right)X^x} \qw & \qw & \multigate{1}{X^{-x}\left(Z^{-g} \otimes Z^b\right)} & \qw \\
\lstick{\mbox{Readout register: } \bra{0}} & \ghost{X^p Z^q} &  \ghost{V} & \ghost{V^\dagger Z^{-q} X^{-p} V} & \ghost{\left(Z^g \otimes Z^a\right)X^x} & \meter\cwx[1] & \ghost{X^{-x}\left(Z^{-g} \otimes Z^b\right)} & \qw \\
& & & & & \control & \cgate{\textrm{add } x} & \cw \gategroup{1}{4}{2}{5}{.7em}{--}
}
\end{align*}
\caption{Randomly compiled implementation of a stabilizer measurement where $V$ is a multi-qudit Clifford operation.
As $V$ is a multi-qudit Clifford operation, $V^\dagger Z^{-q} X^{-p} V$ is a tensor product of single-qudit Weyl operators.
The random Weyl operators are specified by $p, q, x \in \bb{Z}_d^n$, $g \in \bb{Z}_d^{n - m}$, and $a, b \in \bb{Z}_d^m$.}
\label{fig:HigherRankMeasurementsRC}
\end{figure}

While \cref{fig:HigherRankMeasurements} is valid for any unitary $V$, we typically require some structure to ensure that the random gates that are added are simpler than the gate being twirled, and thus do not significantly add to the complexity (and typically the total error) of the circuit.
We now consider the special case where $V$ is a Clifford unitary, that is, when $V W V^\dagger$ is proportional to a Weyl operator for any $W \in \bb{PW}_{d, n}$.
In \cref{fig:HigherRankMeasurementsRC}, we show how to implement \cref{fig:HigherRankMeasurements} with randomized compiling so that Markovian implementations of a Clifford unitary operation and the measurement are reduced to the ideal Clifford channel followed by a uniform stochastic measurement.
Considering first the Clifford operation $V$, from Ref.~\cite{Wallman2016a}, averaging $p, q \in \bb{Z}_d^n$ gives the effective channel
\begin{align}
    \Theta_{\rm RC}(\mc{V}) = \Lambda \mc{V},
\end{align}
where $\Lambda$ is a stochastic Weyl channel, that is, a channel of the form~\cite{fukuda2006}
\begin{align}\label{eq:LambdaWeyl}
    \Lambda &= \sum_{x, z \in \bb{Z}_d^n} \mu(x, z) \mc{X}^x \mc{Z}^z 
\end{align}
where $\mu : \bb{Z}_d^{2n} \to [0, 1]$ is a probability distribution.
Combining this with \cref{thm:RCMeas}, when we average over all the random gates in \cref{fig:HigherRankMeasurementsRC}, we obtain the channel
\begin{align}\label{eq:IndirectRC}
    \Theta_{\rm RC}(\mc{M}_{U, \bb{H}_{\rm meas}} \mc{V}) &=\Theta_{\rm RC}(\mc{M}_{U, \bb{H}_{\rm meas}})\Theta_{\rm RC}(\mc{V})\notag\\
    &=\left(\sum_{a, b, k \in \bb{Z}_d^m} \mc{T}_{a, b} \otimes \cketbra{k + b}{k + a} \otimes \cket{k}\right)\Lambda \mc{V}
\end{align}
where the $\mc{T}_{a, b}$ are unnormalized stochastic channels.
Moreover, as the $\mc{T}_{a, b}$ are twirled by the Weyl group, they are also of the form
\begin{align}\label{eq:TabWeyl}
    \mc{T}_{a, b} = \sum_{\alpha, \beta \in \bb{Z}_d^{n - m}} \nu_{a, b}(\alpha, \beta) \mc{X}^\alpha \mc{Z}^\beta
\end{align}
where each $\nu_{a, b} : \bb{Z}_d^{n-m} \to [0, 1]$ is a subnormalized probability distribution that together add to a normalized probability distribution.
Substituting \cref{eq:LambdaWeyl,eq:TabWeyl} into \cref{eq:IndirectRC} and using the fact that the product of Weyl channels is a Weyl channel gives
\begin{align}\label{eq:IndirectMeasurementGeneralClifford}
    \Theta_{\rm RC}(\mc{M}_{U, \bb{H}_{\rm meas}} \mc{V}) &= \sum_{a, b, k \in \bb{Z}_d^m} \left(\Lambda_{a, b} \otimes \cketbra{k + b}{k + a} \right) \mc{V} \otimes \cket{k},
\end{align}
for some subnormalized stochastic channels $\Lambda_{a, b}$.
When $V$ is a Clifford operation that is also of the special form in \cref{eq:ConditionalProjectors}, \cref{eq:IndirectMeasurementGeneralClifford} simplifies to performing the ideal projective measurement and then applying a uniform stochastic channel.
While we focused on measurements in the computational basis throughout this paper, these results can be generalized to measurements in another basis by re-defining the random $Z$ and $X$ gates relative to the measurement basis.

\subsection{Two-qubit example}\label{sec:ReturnOfTheIndMeasEx}

For concreteness, we return to the simple two-qubit example in \cref{sec:indMeasEx} and show how it is transformed by randomized compiling.
Using the fact that $Z$ is diagonal and use the identities $CX (X \otimes I) CZ = X \otimes X$ and $CX (I \otimes Z) CX = Z \otimes Z$ to determine the correction gates, so that the randomly compiled implementation of \cref{sec:indMeasEx} is as follows, where the $x_j$ and $z_j$ are bits.
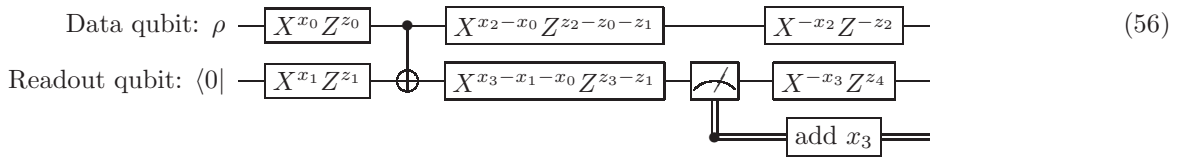
\begin{align}\label{eq:fullRCIndMeasEx}
\Qcircuit @C=1em @R=.7em {
\lstick{\mbox{Data qubit: } \rho} & \gate{X^{x_0} Z^{z_0}} & \ctrl{1} & \gate{X^{x_2 - x_0} Z^{z_2 - z_0 - z_1}} & \qw & \gate{X^{-x_2} Z^{-z_2}} & \qw \\
\lstick{\mbox{Readout qubit: } \bra{0}} & \gate{X^{x_1} Z^{z_1}} & \targ & \gate{X^{x_3 - x_1 - x_0} Z^{z_3 - z_1}} & \meter\cwx[1] & \gate{X^{-x_3} Z^{z_4}} & \qw \\
 & &  & & \control & \cgate{\textrm{add } x_3} & \cw
}
\end{align}
As we will only consider a noisy implementation of the controlled NOT gate, we will simplify the above by compiling adjacent single-qubit operations, however, this compilation would result in a different noise process if the single-qubit operations are implemented imperfectly.
We simplify \cref{eq:fullRCIndMeasEx} by applying the first $Z^{z_1}$ to the input state and canceling the random gates where possible by commuting them through single-qubit circuit elements, noting that $\mc{X}$ and $\mc{Z}$ commute as channels.
Because the computational basis measurement is ideal, the random $X$ and $Z$ gates applied before and after it leave it unchanged. These gates should not generally be omitted when the measurement itself is noisy, but in our example the noise is coming solely from a multi-qubit gate, so they have no effect.
These simplifications give the following circuit.
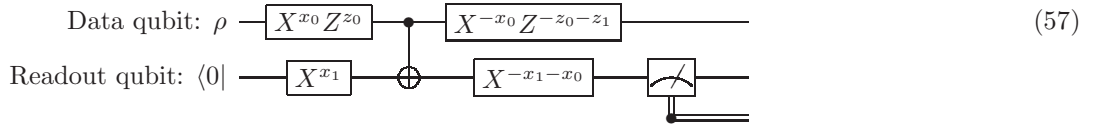
\begin{align}\label{eq:simplifiedRCIndMeasEx}
\Qcircuit @C=1em @R=.7em {
\lstick{\mbox{Data qubit: } \rho} & \gate{X^{x_0} Z^{z_0}} & \ctrl{1} & \gate{X^{- x_0} Z^{- z_0 - z_1}} & \qw & \qw& \qw \\
\lstick{\mbox{Readout qubit: } \bra{0}} & \gate{X^{x_1}} & \targ & \gate{X^{- x_1 - x_0}} & \meter\cwx[1] & \qw & \qw \\
 & &  & & \control & \cw & \cw
}
\end{align}
While \cref{eq:simplifiedRCIndMeasEx} is equivalent to simply randomly compiling the controlled NOT, randomized compiling has not been analyzed in the context of non-destructive measurements. \Cref{thm:RCMeas} establishes that the complete indirect measurement with general errors will be reduced to a uniform stochastic measurement as in \cref{eq:IndirectMeasurementGeneralClifford}, however, the only simple error models for measurements are confusion matrices, which do not fully capture the dynamics of a non-destructive measurement.

We therefore restrict attention to the well-motivated example of an over-rotation gate during an indirect measurement and show that RC helps for this explicit example.
When the controlled not in \cref{eq:simplifiedRCIndMeasEx} is implemented imperfectly using the $V_\phi$ gate defined in \cref{eq:Vdef}, the overall channel is equivalent to the following.
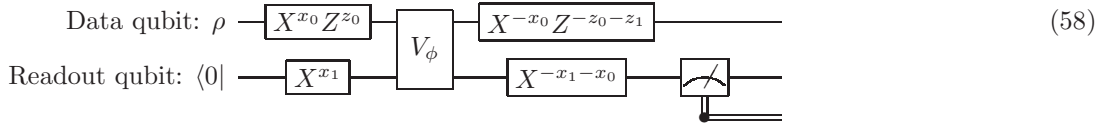
\begin{align}
\Qcircuit @C=1em @R=.7em {
\lstick{\mbox{Data qubit: } \rho} & \gate{X^{x_0} Z^{z_0}} & \multigate{1}{V_\phi} & \gate{X^{- x_0} Z^{- z_0 - z_1}} & \qw & \qw& \qw \\
\lstick{\mbox{Readout qubit: } \bra{0}} & \gate{X^{x_1}} & \ghost{V_\phi} & \gate{X^{- x_1 - x_0}} & \meter\cwx[1] & \qw & \qw \\
 & &  & & \control & \cw & \cw
}
\end{align}
As $V_\phi$ commutes with $Z \otimes I$ and $I \otimes X$, we obtain the reduced circuit
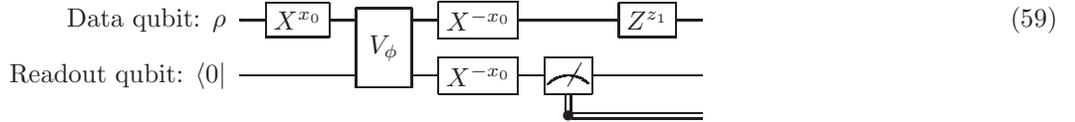
\begin{align}
\Qcircuit @C=1em @R=.7em {
\lstick{\mbox{Data qubit: } \rho} & \gate{X^{x_0}} & \multigate{1}{V_\phi} & \gate{X^{- x_0}} & \qw & \gate{Z^{z_1}}& \qw \\
\lstick{\mbox{Readout qubit: } \bra{0}} & \qw & \ghost{V_\phi} & \gate{X^{- x_0}} & \meter\cwx[1] & \qw & \qw \\
 & &  & & \control & \cw & \cw
}
\end{align}
Setting aside the $Z^{z_1}$ gate (which will average to a dephasing channel by \cref{lem:Dephasing}) to begin, we have four Kraus operators arising from two measurement outcomes and two values of $x_0$
These Kraus operators are given by
\begin{align}
    L_{k,x_0} = \begin{cases}
    \left(\ketbra{x_0} + \cos\phi\ketbra{1+x_0}\right) \otimes \ket{k} & k = x_0 \\
    -i\sin\phi\ketbra{1+x_0}\otimes\ket{k}& k \neq x_0.
    \end{cases}
\end{align}
Averaging over $z_1\in\bb{Z}_2$ removes any off-diagonal terms introduced by conjugating by these Kraus operators, so that effectively we have 6 Kraus operators for the final channel.
We can combine like terms to reduce this down to 4 Kraus operators of the form
\begin{align}
    \tilde{K}_{k',k} = \sqrt{|p(k,k')|}\ketbra{k'}\otimes\ket{k},
\end{align}
where $k'$ indicates which ideal projector acts on the state space, $k$ is the result written to the readout register, and $p(k,k')$ is a coefficient representing the probability of observing $k$ and applying $k'$.
After this simplification, the final set of Kraus operators can be written
\begin{align}
    \tilde{K}_{0,0} &= (1-i\sin\phi)\ketbra{0}\otimes\ket{0}\notag\\
    \tilde{K}_{0,1} &= \cos\phi\ketbra{1}\otimes\ket{0}\notag\\
    \tilde{K}_{1,0} &= \cos\phi\ketbra{0}\otimes\ket{0}\notag\\
    \tilde{K}_{1,1} &= (1-i\sin\phi)\ketbra{0}\otimes\ket{0}.
\end{align}
This noise can thus be condensed into a confusion matrix representation
\begin{align}
    \frac{1}{2}\begin{pmatrix}
        1+\sin^2\phi & \cos^2\phi\\
        \cos^2\phi & 1+\sin^2\phi
    \end{pmatrix}.
\end{align}
Adding additional errors to the CZ gate may change the confusion matrix and/or add a stochastic Pauli error channel acting on the state space.
In \cref{sec:appendix-noisy-ind-meas}, we show how this example generalizes to higher dimensions and also indirect measurements of multi-qudit operators.

\section{Conclusion}

Measurements are a fundamental part of quantum computations, and are a major source of intractable noise in modern quantum systems.
In this paper, we showed that the common method for modeling noisy measurements is inadequate to capture the dynamics in a realistic system.
We proposed a compilation method which addresses the intractability of measurement errors, rendering them easier to mitigate and model, and showed that applying our technique reduces measurement errors to a form which is commonly used to model measurement errors, that is, to probabilistic misreporting of measurement outcomes and a stochastic channel on the unmeasured qudits.
Our method applies to non-destructive and indirect measurements, making it fully generalizable to any case where measurements could be used in a quantum computation, and does not introduce significant overhead in either the number or gates or number of shots required.

\section{Acknowledgements}

We would like to thank Samuele Ferracin, Matthew Graydon, and Emily Wright for helpful discussions. 
This research was supported by the U.S. Army Research Office through grant W911NF-21-1-0007, the Canada First Research Excellence Fund, the Government of Ontario, and the Government of Canada through NSERC.

\bibliographystyle{apsrev4-2}
\bibliography{library}
\appendix

\section{Indirect measurements with noisy Clifford operations}\label{sec:appendix-noisy-ind-meas}

Noise can arise in any or all of the component operations used to implement an indirect measurement.
In this subsection, we look specifically at the effect of an over- or under-rotation of the coupling operation used to implement an indirect Weyl measurement.

An explicit form for a circuit which performs an indirect measurement of a Weyl operator $A\in\bb{PW}_{d, n}$ is given in \cref{fig:WeylMeas}, where $F$ is the quantum Fourier transform,
\begin{align}
    F &= \frac{1}{\sqrt{d}}\sum_{a,b\in\bb{Z}_d}\chi_a(b)\ketbra{a}{b},
    \label{eq:Fourier}
\end{align}
and a controlled-$A$ gate is given by
\begin{align}
    \hat{A}= \sum_{j\in\bb{Z}_d} \ketbra{j} \otimes A^j.
    \label{eq:ControlledWeyl}
\end{align}

\begin{figure}[t!]
\begin{align*}
\Qcircuit @C=1em @R=.7em {
\lstick{\mbox{State register}} & \qw & \qw &  \gate{A} & \qw & \qw & \qw \\
\lstick{\mbox{Readout register: } \bra{0}} & \qw & \gate{F} & \ctrl{-1} & \gate{F^\dagger} & \meter & \cw \\
}
\end{align*}
\caption{Circuit which performs an indirect measurement of a Weyl operation $A\in\bb{W}(\bb{H}_s)$, where the quantum Fourier transform is defined in \cref{eq:Fourier} and controlled-Weyl gates are defined in \cref{eq:ControlledWeyl}. The measurement is performed in the computational basis.}
\label{fig:WeylMeas}
\end{figure}
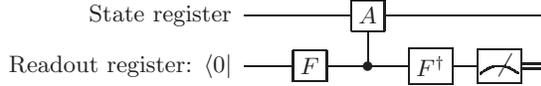

We assume that the Fourier and inverse Fourier operations are performed perfectly and implement the controlled-$A$ operation as an imperfect rotation.
Recall that the ideal controlled-$A$ operation is
\begin{align}
    \hat{A} = \sum_{j\in\bb{Z}_d} \ketbra{j}\otimes A^{j}.
\end{align}
We now define an operation which can represent an over- or under-rotation of the applied operation,
\begin{align}
    \hat{A}(t) \equiv \hat{A}^t = \sum_{j\in\bb{Z}_d} \ketbra{j}\otimes A^{jt},
\end{align}
where the ideal operation is applied when $t=1$, that is $\hat{A}(1)=\hat{A}$.
We will assume that $t$ is close to $1$ and let $t=1+\epsilon$. Then
\begin{align}
    A^{jt} = A^{j(1+\epsilon)}=A^j(A^j)^\epsilon.
    \label{eq:Ajt_factorization}
\end{align}
For small $\epsilon$, we can expand $B^\epsilon$ as
\begin{align}
    B^\epsilon = I+\epsilon \rm{ln}(B)+\mc{O}(\epsilon^2).
\end{align}
Applying this expansion to \cref{eq:Ajt_factorization},
\begin{align}
    A^{jt} &= A^j[I+\epsilon \rm{ln}(A^j)+\mc{O}(\epsilon^2)]\\
    &\approx  A^j+\epsilon jA^j\rm{ln}(A).
    \label{eq:Ajt_expansion}
\end{align}

To get the complete expression for the Kraus operators of an indirect measurement with noise of the described form, we fold the Fourier transform into the input state to get an initial state

\begin{align}
    F\ket{0} = \frac{1}{\sqrt{d}}\sum_{a\in\bb{Z}_d}\ket{a} \equiv \ket{+},
    \label{eq:prepPlus}
\end{align}
and fold the inverse Fourier transform into the computational basis measurement, to get an effective destructive measurement on the readout space with Kraus operators
\begin{align}
    \mc{M}_{I,\bb{H}_r}F^\dagger &= \{\bra{k}F^\dagger: k\in\bb{Z}_d\}\\
    &= \{\frac{1}{\sqrt{d}} \sum_{b\in\bb{Z}_d} \bar{\chi}_b(k)\ketbra{k}{b}: k\in\bb{Z}_d\}
    \label{eq:fourierMeas}
\end{align}
acting on the state register.
Combining the input state (\cref{eq:prepPlus}), interaction operation, $\hat{A}^t$, and effective measurement (\cref{eq:fourierMeas}), we get Kraus operators for the noisy measurement of the form
\begin{align}
    M_k(t) = \frac{1}{d}\sum_{j\in\bb{Z}_d}\bar{\chi}_k(j)A^{jt}.
    \label{eq:noisyKrausMeas}
\end{align}
Note that for the remainder of this section we omit the readout register, treating this as a destructive measurement, for ease of book-keeping.
However, the following analysis generalizes trivially to a non-destructive measurement.
Substituting \cref{eq:Ajt_expansion} into the Kraus operators in \cref{eq:noisyKrausMeas}, we get
\begin{align}
    M_k(1+\epsilon) &\approx \frac{1}{d}\sum_{j\in\bb{Z}_d}\bar{\chi}_k(j)[A^j+\epsilon jA^j\rm{ln}(A)] \notag\\
    &\approx \pi_k(A) + \frac{\epsilon\rm{ln}(A)}{d}\sum_{j\in\bb{Z}_d}\bar{\chi}_k(j)jA^j,
\end{align}
where $\pi_k(A) = M_k(1)$ are the projectors on the state space for the ideal measurement of $A$.
Ideally, a measurement of $A$ with outcome $k$ projects fully onto the $k^{th}$ eigenspace of $A$ so that for a set of eigenvectors $\{\ket{\psi_b}: b\in\bb{Z}\}$ of $A$, $\bra{\psi_b}\pi_k(A)\ket{\psi_b}=\delta_{b,k}\braket{\psi_b}$.
If the second term is not proportional to $\pi_k(A)$, this measurement can leave some of the state outside of the $k^{th}$ eigenspace of $A$.
To see this, let
\begin{align}
    J_k = \sum_{j\in\bb{Z}_d}\bar{\chi}_k(j)jA^j\rm{ln}(A),
\end{align}
so that we can describe the second term more succinctly, and observe that the terms resulting from applying $M_k(t)$ to a state $\bar{\rho}$ will be proportional to:
\begin{enumerate}
    \item $\pi_k(A)\bar{\rho}\pi_k(A)^\dagger$ (which is in the expected eigenspace)
    \item cross terms $J_k\bar{\rho}\pi_k(A)^\dagger$ and $\pi_k(A)\bar{\rho}J_k^\dagger$, and
    \item $J_k\bar{\rho}J_k^\dagger$.
\end{enumerate}
If $J_k$ is proportional to a projector (or a combination of projectors) other than $\pi_k(A)$, the cross terms allow superpositions between the cospaces to persist beyond the measurement, and $J_k\bar{\rho}J_k^\dagger$ allows terms outside of the $k^{th}$ eigenspace to persist as well.

We now show that $J_k$ does not implement a complete projection onto the $k^{th}$ eigenspace, i.e. $J_k$ is not proportional to $\pi_k(A)$.
We can define an eigenbasis of $A$ with eigenvalues $\{\chi_b(1): b\in\bb{Z}_d\}$ by selecting a $+1$ eigenvector $\ket{\psi_0}$ of $A$ for $b=0$ and letting the remaining eigenvectors be $\ket{\psi_b}= T_b\ket{\psi_0}$, where $AT_b=\chi_b(1)T_bA$.
To show that $J_k$ does not project onto the $k^{th}$ eigenspace, we need to demonstrate that $J_k\ket{\psi_b}$ is non-zero for some $b\neq k$.
If $\ket{\psi_b}$ is an eigenvector of $A$ with eigenvalue $\chi_b(1)$, then $\ket{\psi_b}$ is also an eigenvector of $\rm{ln}(A)$ with eigenvalue $\rm{ln}(\chi_b(1))$.
So, recalling \cref{eq:CyclicCharacters}, we have
\begin{align}
    J_k\ket{\psi_b} &= \sum_{j\in\bb{Z}_d} \bar{\chi}_k(j)j\rm{ln}[\chi_{b}(1)]\chi_b(1)^j\ket{\psi_b}\\
    &= \frac{2\pi i b}{d}\sum_{j\in\bb{Z}_d} \bar{\chi}_k(j)\chi_b(j)j\ket{\psi_b}.
\end{align}
Using the series identity
\begin{align}
    \sum_{i=0}^n k x^k = x\frac{1 - (n+1) x^n + n x^{n+1}}{(1 - x)^2}
\end{align}
and \cref{eq:CyclicCharacters}, it can be shown that
\begin{align}
    \frac{2\pi i}{d}\sum_{j\in\bb{Z}_d} \chi_1(j)j = \pi \left[\cot\left(\frac{\pi}{d}\right)-i\right].
\end{align}
Therefore for $k=0$ and $b=1$, we have
\begin{align}
    J_0\ket{\psi_1} &= \pi \left[\cot\left(\frac{\pi}{d}\right)-i\right]\ket{\psi_1},
\end{align}
which is always nonzero.
As $0$ and $1$ are in $\bb{Z}_d$ for any $d>2$, there will always be an outcome which can leave a portion of the state outside of the expected eigenstate after measurement for any meaningful qudit system.
We can therefore conclude that an over- or under- rotation of the coupling gate for an indirect Weyl measurement can leave a portion of the state outside of the eigenspace associated with the measured outcome.
Note that the asymmetry between, e.g., outcomes $k=0$ and $k=1$ arises as a result of the asymmetry in the error model; if the state of the system prior to the measurement is a $+1$ eigenvector of $A$, then $\hat{A}(t)$ leaves it unchanged for any value of $t$.

\end{document}